\documentclass[12pt]{amsart}
\pdfoutput=1

\usepackage{cite}

\usepackage{lineno,hyperref}
\modulolinenumbers[5]
\usepackage{epstopdf}
\usepackage{pifont}
\usepackage{latexsym}
\usepackage{graphics}
\usepackage{graphicx}
\usepackage{float}
\usepackage[dvips]{xy}
\xyoption{all}
\usepackage{ifpdf}
\usepackage{multirow}
\usepackage{amssymb, amsthm, amsmath, mathrsfs, stmaryrd }
\usepackage{dsfont}
\usepackage{hhline}
\usepackage{centernot}
\usepackage{verbatim, comment, color}
\usepackage{enumerate}
\usepackage{setspace}
\usepackage{lipsum}
\usepackage{subcaption}

\usepackage{mathtools}

\newtheorem{theorem}{Theorem} 
\newtheorem{proposition}[theorem]{Proposition}

\newcommand{\ba}{\begin{align}}
\newcommand{\ea}{\end{align}}  

\newcommand{\be}{\begin{equation}}

\newcommand{\ee}{\end{equation}}
\newcommand{\bea}{\begin{eqnarray}}
\newcommand{\eea}{\end{eqnarray}}
\newcommand{\barr}{\begin{array}}
\newcommand{\earr}{\end{array}}
\newcommand{\bn}{\begin{enumerate}}
\newcommand{\en}{\end{enumerate}}
\newcommand{\bi}{\begin{itemize}}
\newcommand{\ei}{\end{itemize}}
\newcommand{\bbbm}{\begin{pmatrix}}
\newcommand{\eeem}{\end{pmatrix}}

\newcommand{\cP}{{\cal P}}

\newcommand{\cX}{{\cal X}}

\newcommand{\R}{{\mathbb R}}
\newcommand{\N}{{\mathbb N}}
\newcommand{\E}{\mathbb{E}}

\newcommand{\ignore}[1]{}{}

\newcommand{\nn}{\nonumber}

\newcommand{\q}{\quad}

\newcommand{{\QED}}{{\hfill QED} \smallskip}

\renewcommand{\subset}{\subseteq}

\renewcommand{\phi}{\varphi}
\newcommand{\cal}{\mathcal}

 \definecolor{darkspringgreen}{rgb}{0.09, 0.45, 0.27} 
 \definecolor{darkgray}{rgb}{0.66, 0.66, 0.66}
\numberwithin{equation}{section}
\numberwithin{theorem}{section}

\begin{document}
\title
[Multi-period vectorial martingale optimal transport]
{Replication of financial derivatives under extreme market models given marginals
} 

\thanks{The author wishes to express gratitude to the Korea Institute of Advanced Study (KIAS) AI research group and the director Hyeon, Changbong for their hospitality and support during his stay at KIAS in 2023, where parts of this work were performed.
}
\date{\today}

\author{Tongseok Lim}
\address{Tongseok Lim: Mitchell E. Daniels, Jr. School of Business \newline  Purdue University, West Lafayette, Indiana 47907, USA}
\email{lim336@purdue.edu}
\onehalfspacing

\begin{abstract}
The Black-Scholes-Merton model is a mathematical model for the dynamics of a financial market that includes derivative investment instruments, and its formula provides a theoretical price estimate of European-style options. The model's fundamental idea is to eliminate risk by hedging the option by purchasing and selling the underlying asset in a specific way, that is, to replicate the payoff of the option with a portfolio (which continuously trades the underlying) whose value at each time can be verified. One of the most crucial, yet restrictive, assumptions for this task is that the market follows a  geometric Brownian motion, which has been relaxed and generalized in various ways.

The concept of robust finance revolves around developing models that account for uncertainties and variations in financial markets. Martingale Optimal Transport, which is an adaptation of the Optimal Transport theory to the robust financial framework, is one of the most prominent directions. In this paper, we consider market models with arbitrarily many underlying assets whose values are observed over arbitrarily many time periods, and demonstrates the existence of a portfolio sub- or super-hedging a general path-dependent derivative security in terms of trading European options and underlyings, as well as the portfolio replicating the derivative payoff when the market model yields the extremal price of the derivative given marginal distributions of the underlyings. In mathematical terms, this paper resolves the question of dual attainment for the multi-period vectorial martingale optimal transport problem.
\end{abstract}

\maketitle


\section{Introduction}
The Black-Scholes-Merton model, often referred to as the Black-Scholes model, is a mathematical framework for calculating the theoretical pricing of options and other derivatives \cite{bs73}. It was developed by economists Fischer Black and Myron Scholes in collaboration with mathematician Robert C. Merton in the early 1970s. 

The Black-Scholes-Merton model revolutionized the way options are priced in financial markets by providing a mathematical formula to determine the fair value of options based on various factors, and their pricing model has become a fundamental tool in options trading and valuation. The Black-Scholes-Merton model combines concepts from mathematics, statistics, and economics to provide a rigorous framework for understanding option pricing and risk management. As a result, it influenced subsequent research in financial economics and inspired the development of more sophisticated models and techniques \cite{DKSY23, DS94, DS06, FPS00, KS99}.

The idea of replication, also known as the principle of no-arbitrage, is arguably the most fundamental concept underlying the Black-Scholes-Merton model. It states that the value of an option can be recreated by establishing a risk-neutral portfolio of the underlying asset in such a way that the portfolio replicates the option's cash flows and payoffs. This idea has influenced the development of the efficient market hypothesis and has had a profound impact on the study of financial economics. 

Despite its importance, it is worth noting that the Black-Scholes-Merton model has certain limitations, such as the assumption of the market following a geometric Brownian Motion and the assumption of continuous trading. These assumptions may not hold in real-world markets, and there have been subsequent models and refinements that address some of these limitations \cite{cko, DFS03, FPSS11, GH14, Sc17}.

The concept of robust finance revolves around the development of models and approaches that account for market uncertainties. It recognizes the limitations of traditional financial models, which presume exact and accurate information about the market and its underlying assets, and instead attempts to construct frameworks that can adapt to the inherent uncertainties and risks in real-world financial systems \cite{BRS21, D16, H17, Ho11, npx23, Obloj}.

One of the most renowned and prospering directions is the theory of Martingale Optimal Transport (MOT), which is an adaptation of the Optimal Transport (OT) theory to the robust financial framework. The OT theory, also known as the theory of Monge-Kantorovich transportation, is a mathematical framework that deals with the problem of efficiently moving mass from one distribution to another. OT theory seeks to find the optimal way to transport one distribution of mass to another while minimizing the cost or distance associated with the transportation. The theory originated in the 18th century with the works of mathematicians Gaspard Monge and Leonid Kantorovich, and has since been developed and extended by various researchers. The importance of OT theory lies in its wide range of applications across various disciplines, including mathematics, statistics, economics, physics and computer science \cite{ccg16, cfg10, cghh17, cghp21, fkm11, gm, p12, Sa15, Vi09}.

The theory of Martingale Optimal Transport is an extension of OT theory that combines it with the concept of martingales from probability theory. By combining OT theory with the concept of martingales, MOT extends the applicability of optimal transportation models to dynamic and stochastic settings. It provides a valuable tool for understanding and solving transportation problems under uncertainty, with applications in various fields including finance, risk management, stochastic control, and data analysis \cite{bch, bj, cot19, ds1, GaHeTo11, GKL2, GKL3, GKP21, GKP22, gtt1, gtt2, KX22}. 

One of the most prominent applications of MOT is in mathematical finance and option pricing. It provides a framework for modeling and valuing derivative securities in the presence of uncertainty and stochastic dynamics. MOT-based approaches can be used to analyze optimal hedging strategies and calculate prices and risk measures for options and other financial derivatives \cite{BeHePe11, eglo, Ho98, HoKl12, HoNe11}.

The OT and MOT  are infinite-dimensional linear programming problems, hence the problem has a dual programming problem. In MOT, the dual problem has an important interpretation in terms of determining the best sub- or super-hedging portfolio against a derivative security payout liability. As a result, one of the most significant concerns in MOT theory has been the topic of {\em dual attainment}, which refers to whether the dual problem attains an appropriate solution \cite{blo, bnt, CKPS21, Ke84}.

While the dual attainment has been well established to be affirmative in OT, researchers discovered that the added term reflecting trading strategy appearing in the dual of MOT makes establishing the dual attainment for MOT theory far more subtle. Specifically, unlike the OT theory, the dual attainment in MOT is highly sensitive to the spatial dimension, which represents the number of assets in the financial market. Because many derivative instruments are traded in the market and their payoffs are dependent on the values of various underlying assets, it is critical to understand the dual problem and its solution in a multidimensional context, i.e., a market with multiple assets as well as derivatives depending on them. As a result, efforts have been made to comprehend the dual attainment in the higher-dimensional case \cite{d18-1, d18, dt19, GKL2, Lim23, os17}, but a thorough knowledge of this question appears to be far from its completion. In particular, to the best of the author's knowledge, all research articles on the dual attainment of  multidimensional MOT assume two future maturities. In OT, this means that there are two distributions for which we attempt to move mass efficiently from one to the other. While this is a natural and sufficiently general setting in OT, it represents a significant restriction in MOT because it indicates that the reward of the derivative security under consideration can only depend on the values of the underlyings at two future maturities. It is desirable, both theoretically and practically, to investigate the martingale transport theory of price path-dependent derivative instruments whose reward can depend on arbitrarily many assets and maturities.

The goal of this paper is to establish the dual attainment of the martingale optimal transport problem over an arbitrary number of time periods and assets. We show that there exists a portfolio sub- or super-hedging a  path-dependent derivative security in terms of trading European options and underlyings, such that the portfolio replicates the derivative payoff when the market model yields the derivative's extremal price given marginal distributions of the underlyings. We contend  this is fundamental and relevant given that the price path of many underlying assets over time frequently affects financial instruments and their payouts.

This paper is organized as follows. In Section \ref{MTintro}, we introduce the multi-period vectorial martingale optimal transport problem. In Section \ref{contribution}, we discuss duality and the dual attainment result, whose proof is then provided in Section \ref{proof}.

\section{Martingale transport problem with multi assets and periods}\label{MTintro}

Consider the asset price processes $(X_{t,i})_{t \ge 0}$  in the market indexed by $i \in [d] := \{1,2,...,d\}$ representing the $i$th underlying asset. We will not assume in this paper that the joint probability law of the underlyings, often called as the market model, is known, because we cannot determine this joint law from market information. On the other hand, using a standard reasoning by Breeden and Litzenberger \cite{bl78}, we will suppose that the market can witness the distribution of each price at each fixed maturity $t > 0$, denoted by ${\rm  Law}(X_{t, i}) \in \cP(\R)$, where $\cP(\cX)$ denotes the set of all probability distributions over $\cX$. We consider an arbitrary number of finite maturity times $0 < T_1 < T_2 < ... < T_N$, denote $X_{t,i}:= X_{T_t, i}$ and $X_{t} =(X_{t, 1},...,X_{t, d})$, and assume that $(X_t)_{t \in [N]}$ is a $\R^d$-valued martingale, which is generally assumed in the financial literature through the concept of the risk-neutral probability. Now according to the above consideration, we do not assume that the market model ${\rm Law}(X_t)_t \in \cP(\R^{Nd})$ is known, but only the $Nd$-number of marginal distributions $\mu_{t,i}:= {\rm Law}(X_{t,i}) \in \cP(\R)$ are known and fixed. This leads us to consider the {\em Vectorial Martingale Optimal Transport} (VMOT) problem: Assume the marginals $(\mu_{t,i})_{t,i}$ have finite first moment, write $\mu_t = (\mu_{t,1},...,\mu_{t,d})$, $ \mu = (\mu_1,...,\mu_N)$, and $X = (X_1,...,X_N)$. We consider the space of {\em Vectorial Martingale Transports} from $\mu_t$ to $ \mu_{t+1}$, $t=1,...,N-1$, defined as follows:
\begin{align}\label{VMT}
{\rm VMT}(\mu) := \{ \pi \in \cP(\R^{Nd}) \ | \  &\pi = {\rm Law} (X), \ \E_\pi[X_{t+1} | X_t]=X_t, \\
& {\rm Law}(X_{t,i}) = \mu_{t,i} \, \text{ for all } t \in [N],\, i \in [d] \}. \nn
\end{align}
Given a {\em cost function} $c : \R^{Nd} \to \R$, we define the VMOT problem as
\begin{align}\label{VMOT}
\max / \, {\rm minimize } \ \ \E_\pi [c(X)] \ \text{ over } \ \pi \in {\rm VMT}(\mu).
\end{align}
In finance, the function $c$ is naturally interpreted as a  derivative security whose payoff $c(X)$ is fully determined at the terminal maturity $T_N$ by the price path $X = (X_t)_t$ of the $d$-number of underlyings. In this case, $\E_\pi [c(X)]$ can be regarded as a fair price for the derivative security $c$ under the market model $\pi$. Because $\pi$ cannot be observed in the market, we must take into account all feasible models ${\rm VMT}(\mu)$ which are consistent with the marginal information $\mu = (\mu_{t,i})_{t,i}$. With this knowledge, the maximum and minimum values in \eqref{VMOT} can be interpreted as the upper and lower price bounds for the derivative security $c$, respectively.

The VMOT problem is distinguished from the (ordinary) optimal transport problems by the martingale constrant $ \E_\pi[X_{t+1} | X_t]=X_t$ for all $t=1,...,N-1$, which necessitates that every pair of marginals $\mu_{t,i}, \mu_{t+1, i}$ must be in {\em convex order}\,:
\[ \mu_{t,i} \preceq_c \mu_{t+1, i} \ \ \text{if} \  \ \mu_{t,i}(f) \le \mu_{t+1, i}(f) \ \ \text{for every convex function } f \text{ on } \R,
\]
where $\mu(f) := \int f(x) \mu(dx)$, in which case ${\rm VMT}(\mu) \neq \emptyset$ and vice versa, as shown by Strassen \cite{St65}. Thus we will assume $\mu_{t,i} \preceq_c \mu_{t+1, i}$ for all $t \le N-1$ and $i \in [d]$.

The VMOT problem belongs to the class of {\em infinite-dimensional linear programming}, hence the problem has its {\em dual programming} problem. When the primal  problem \eqref{VMOT} is a minimization problem, its dual problem is given by 
 \begin{equation}\label{dualproblem}
 \sup_{(\phi,h) \in \Psi} \mu(\phi),
 \end{equation}
where $\phi = (\phi_1,...,\phi_N)$, $\phi_t = (\phi_{t,1},...,\phi_{t, d})$, $\phi_{t,i} : \R \to \R \cup \{-\infty\}$, and $\mu(\phi) := \sum_{t,i} \mu_{t,i}(\phi_{t,i})$. Meanwhile, $h = (h_1,...,h_{N})$ with the convention $h_N \equiv 0$, $h_t = (h_{t,1},...,h_{t,d})$ where $h_{t,i} : \R^{td} \to \R$ is a function of $(X_1,...,X_t)$. Finally, $(\phi,h) \in \Psi$ means that $\phi_{t,i} \in L^1( \mu_{t,i})$, $h_{t,i}$ is bounded, and the following inequality holds:
\begin{align}
\label{ptwiseineq} 
\sum_{t=1}^N \sum_{i=1}^d  \phi_{t,i} (x_{t,i}) + h_{t,i}(x_1,...,x_t) (x_{t+1, i} - x_{t,i}) \le c(x),
\end{align}
where $x=(x_1,...,x_N) \in \R^{Nd}$, $x_t=(x_{t,1},...,x_{t,d}) \in \R^d$ represents a price path. Note that
\be\label{hedgeprice}
\mu(\phi) = \E_\pi \bigg[ \sum_{t=1}^N \sum_{i=1}^d  \phi_{t,i} (X_{t,i}) + h_{t,i}(X_1,...,X_t) (X_{t+1, i} - X_{t,i})\bigg]
\ee
for any $\pi \in {\rm VMT}(\mu)$ and $(\phi,h) \in \Psi$, since $\pi$ has marginals $\{\mu_{t,i}\}_{t,i}$ and the martingale property of $\pi$ implies $ \E_\pi[h_{t,i}(X_1,...,X_t) (X_{t+1, i} - X_{t,i})] = 0$. Now if \eqref{VMOT} is a maximization problem, its dual problem reads $ \inf_{(\phi,h) \in \Psi} \mu(\phi)$, with the inequality \eqref{ptwiseineq} reversed (and $\phi_{t,i}$ taking on their values in $\R \cup \{+\infty\}$).

The dual problem also has a significant interpretation and implication in finance. Let us assume that a financial firm is obligated to pay $c(X)$ at the terminal maturity $N$ for a derivative instrument $c$. To mitigate risk, the firm may consider purchasing or selling European options $\phi_{t,i}$ available from the market, the payoff of which is based solely on the price $X_{t,i}$ at the maturity $t$. Furthermore, the company may consider holding $h_{t,i}$ shares of the $i$th asset between the $t$ and $t+1$ maturities, so that its return at time $t+1$ is $h_{t,i}(X_1,...,X_t) \cdot (X_{t+1} -X_t)$. It is worth noting that $h_{t,i}$ is a function of all underlying prices up to time $t$. Then the left hand side of \eqref{ptwiseineq} represents the payout of the hedging portfolio $(\phi, h)$, and the inequality \eqref{ptwiseineq} mandates that the position must subhedge the liability for all possible market realizations $x \in \R^{Nd}$. Having stated that, notice  the dual of the maximization problem in \eqref{VMOT} represents an optimal superhedging problem.

\section{Duality and our contribution}\label{contribution}

The celebrated {\em duality} result asserts, under a mild assumption on $c$ and the marginals, that the primal and dual optimal values coincide (see e.g. \cite{eglo, Z}):
\begin{align}\label{duality}
P(c):&=\inf_{\pi \in {\rm VMT}(\mu)} \E_\pi [c(X)] \\
&= \sup_{(\phi, h) \in \Psi} \mu(\phi)=:D(c). \nn
\end{align}
In addition, the primal problem is known to be attained; there exist an optimizer (VMOT) $\pi$ that yields $\E_\pi [c(X)] = P(c)$. Observe the VMOTs describe extreme market movements in the sense of maximizing or minimizing the fair price of $c$.

Unlike the primal problem, whose attainment can be easily established by standard arguments under a mild condition on $c$,\footnote{This does not imply that the primal optimizers are easy to understand, describe or characterize.} establishing {\em dual attainment}, i.e., proving the existence of a suitable form of dual optimizers, turns out to be a very nontrivial problem, as demonstrated by Brenier's work \cite{br} on optimal transport problems. The situation is worse for the martingale optimal transport problems, where the martingale constraint makes the dual attainment problem even more complex, as shown by \cite{bj, blo, bnt} even for a single asset setup ($d=1$). Interestingly, these works show that the dual attainment problem frequently boils down to establishing convergence of a certain sequence of convex potentials.

The difficulty of the dual attainment problem is problematic not only from a mathematical standpoint, but also from a financial standpoint, because dual optimizers describe how to (sub-/super-)hedge a particular derivative investment most effectively, and furthermore, they often provide the most critical information on the structure of primal optimizers that describe the extreme market models. Due to its significance, numerous literature addresses the dual attainment, the majority of which focus on a single asset $d=1$ and two period setup $N=2$ with a few exceptions, such as  \cite{cot19, nst20, os} which investigated duality in a multi-period with a single asset setup, while \cite{d18-1, d18, dt19, GKL2, Lim23, os17} investigated the structure of vector-valued martingale transports in a two-period setup. The purpose of this paper is to establish the dual attainment of the  martingale optimal transport problem over an arbitrary number of time periods and assets, which is fundamental and relevant given that financial instruments and their payouts frequently depend on the price path of many underlyings over time. Specifically, the goal of this paper is to prove the following result (see Section \ref{proof} for irreducibility.)

\begin{theorem} \label{main}
Let $(\mu_{t,i})_{t \in [N]}$ be an irreducible sequence of marginals on $\R$ for each $i \in [d]$. Let $c(x)$ be a lower semi-continuous cost satisfying $|c(x)| \le \sum_{t,i} v_{t,i}(x_{t,i})$ for some continuous functions $v_{t,i} \in L^1(\mu_{t,i})$. Then there exists a {\em dual optimizer}, that is a pair of function sequences $(\phi, h) = (\phi_{t,i}, h_{t,i})_{t \in [N], i \in [d]}$ that satisfies \eqref{ptwiseineq} tightly in the following pathwise sense (but needs not be in $\Psi$): 
\begin{align}
 \label{ptwiseeq}
 \sum_{t=1}^N \sum_{i=1}^d  \phi_{t,i} (x_{t,i}) + h_{t,i}(x_1,...,x_t) (x_{t+1, i} - x_{t,i}) = c(x) \q \pi-a.s.
\end{align}
for every {\rm VMOT} $\pi$ which solves the minimization problem in \eqref{VMOT}.
\end{theorem}

The $(\phi, h)$ presented in the theorem is known as a {\em dual maximizer} since it is a solution concept to the problem \eqref{dualproblem}, which is dual to the primal minimization problem. The equality \eqref{ptwiseeq} indicates that the portfolio composed of a dual maximizer replicates the derivative security $c$ by yielding the same payout for all feasible price paths $x$  given by any VMOT $\pi$ minimizing the price $\E_\pi[c(X)]$. And \eqref{ptwiseineq} shows the portfolio otherwise subhedges the derivative $c$  for all price paths. Theorem \ref{main} also yields the existence of a dual minimizer such that the portfolio built of a dual minimizer superhedges an upper semi-continuous $c$ by reversing the inequality in \eqref{ptwiseineq}, and it replicates $c$ with respect to the primal maximizers.

We emphasize that a dual optimizer does not necessarily belong to $\Psi$. Studies showed the dual problem \eqref{dualproblem} is generally not attained within the class $\Psi$ even when $(d,N) = (1,2)$ (see \cite{bj, bnt}), unless $c$ satisfies a specific regularity property \cite{blo}. Although $\Psi$ can be seen as a natural domain for the dual problem, it is rather ``narrow" as it lacks suitable compactness due to its infinite dimensionality. This implies that establishing dual attainment is substantially more involved than establishing duality \eqref{duality}, which can usually be derived via standard argument in functional analysis and variational calculus. In summary, a dual optimizer $(\phi, h)$ needs not be in $\Psi$, but it does hold that $\phi_{t,i}$ is real-valued $\mu_{t,i}$-almost surely, and that $h_{t,i}$ is also real valued. Since the marginal distributions $\{\mu_{t,i}\}_{t,i}$ are assumed in the VMOT problem, all of the functions in Theorem \ref{main} are essentially real-valued (and measurable), while no further regularity is imposed a priori.

\section{Proof of Theorem \ref{main}}\label{proof}
When $N=2$, \cite{Lim23} proved Theorem \ref{main}, and one of the most essential ingredients for the proof is the stability of the {\em convex cores} $\{\chi_n\}_n$, which was earlier recognized and proved when $d=1$ in \cite{bnt}. To describe, we need to explain the concept of {\em irreducibility} of probability distributions $\mu \preceq_c \nu$ on $\R$ in convex order. 

Two probabilities (with finite first moment) in convex order 
        $\mu \preceq_{c} \nu$ is called irreducible if $I :=\{x \in \R \, | \, u_{\mu}(x) < u_{\nu}(x)\}$ is a connected interval and $\mu(I)=\mu(\R)$, where $u_{\mu}(x):= \int_\R |x-y| d\mu(y)$ is called the potential function of $\mu$. Note that $I$ is open since $u_{\mu}$ is continuous. In this case, $(I,J)$ is called the domain of $(\mu,\nu)$ where $J$ is the smallest interval satisfying $\nu(J) = \nu(\R)$, that is, $J$ is the union of~$I$ and any endpoints of $I$ that are atoms of $\nu$. Thus in particular, it holds $I={\rm int}(J)$, where ${\rm int}(A)$ is the interior of $A$ and ${\rm conv}(A)$ is the convex hull of $A$. Note that $J$ can be of the form $(a,b]$, $[a,b)$, $(a,b)$ or $[a,b]$; in the first case it holds $\nu(b) > 0$, and in the second $\nu(a) > 0$. And in all cases $I = {\rm int}(J) = (a,b)$. Of course, $I$ and $J$ can be (half)-infinite intervals. 
Roughly speaking, the irreducibility of $\mu \preceq_c \nu$ means that $\nu$ is regularly dispersed from $\mu$. We underline that irreducibility is a natural and generic property that practically any pair of probability distributions on $\R$ in convex order fulfills, and that even if a pair is not irreducible, it can be perturbed arbitrarily small to become irreducible.

Now for irreducible pairs $(\mu_i,\nu_i)_{i \in [d]}$ with domains $(I_i, J_i)$, set $I = I_1 \times ... \times I_d$, $J = J_1 \times ... \times  J_d$ and  $\mu^\otimes = \mu_1 \otimes ... \otimes \mu_d$,  $\nu^\otimes = \nu_1 \otimes ... \otimes \nu_d$. $I$ is an open rectangle in $\R^d$, $J = {\rm int} (I)$, and $\mu^\otimes, \nu^\otimes   \in \cP(\R^d)$ are the product measures of $\mu_i$'s and $\nu_i$'s respectively. Now the following was shown in \cite{bnt} for $d=1$ and in \cite{Lim23} for $d \ge 2$.

\begin{proposition}\label{conv} Let $(\mu_i,\nu_i)_{i \in [d]}$ be  irreducible pairs of probability distributions on $\R$ with domains $(I_i,J_i)_{i \in [d]}$. Let $a \in I$, $C \in \R$. Consider the following class of functions $\Lambda=\Lambda(a,C,\vec\mu,\vec\nu)$ where every $\chi \in \Lambda$ satisfies the following:
\begin{enumerate}
\item $\chi$ is a real-valued convex function on $J$,
\item $\chi \ge 0$ and $\chi(a)=0$,
\item $\int \chi \,d(\nu^\otimes - \mu^\otimes) \le C$.
\end{enumerate}
Then $\Lambda$ is locally bounded in the following sense: for each compact subset $K$ of $J$, there exists $M=M(K)$ such that $\chi \le M$ on $K$ for every $\chi \in \Lambda$. Furthermore, for any sequence $\{\chi_n\}_n$ in $\Lambda$, there exists a subsequence $\{ \chi_{n_j} \}_j$ of $\{ \chi_n\}_n$ and a real-valued convex function $\chi$ on $J$ such that $\lim_{j \to \infty} \chi_{n_j} (x) = \chi(x)$ for all $x \in J$.
\end{proposition}

In Theorem \ref{main}, there are $Nd$ number of marginal distributions $\mu_{t,i}$, indexed by $t \in [N]$ and $i \in [d]$, where $t$ represents the time period and $i$ is the martingale (or financial asset) index. Let $(I_{t,i}, J_{t,i})$ denote the domain of $\mu_{t,i} \preceq_c \mu_{t+1,i}$, $t=1,...,N-1$, with the convention $J_{0,i} := J_{1,i}$ and $I_{N,i} := J_{N-1,i}$. Throughout the proof, keep in mind that the bounding constant $C$ does not depend on $n$.

\begin{proof}[Proof of Theorem \ref{main}]
{\bf Step 1.}  The assumption  $|c(x)| \le \sum_{t, i} v_{t,i}(x_{t,i})$ for continuous $v_{t,i} \in L^1(\mu_{t,i})$ ensures $P(c)=D(c)$ in \eqref{duality} (see e.g. \cite{Z}). Clearly, a dual optimizer exists for $c(x)$ iff so does for $\tilde c (x):= c(x) - \sum_{t,i} v_{t,i}$. Thus by replacing $c$ with $\tilde c$, from now on we will assume that $c \le 0$.

As $P(c)=D(c) \in \R$, we can find an {\em approximating dual maximizer} $(\phi_n, h_n)_{n \in \N}$ which consists of real-valued continuous functions $\phi_{t,i,n} \in L^1(\mu_{t,i})$ and continuous bounded $h_{t,i,n}$ for every $t \in [N]$, $i \in [d]$ and $n \in \N$ (we assume $h_{N,i,n} \equiv 0$),  such that the following duality holds:
\begin{align}
\label{dual} & \sum_{t=1}^N \sum_{i=1}^d  \phi_{t,i,n} (x_{t,i}) + h_{t,i}(\bar x_{t}) \Delta x_{t,i} \le c(x) \le 0,\\
\label{maximizing} &\mu(\phi_n):=  \sum_{t=1}^N \sum_{i=1}^d  \int \phi_{t,i,n} (x_{t,i})\,d\mu_{t,i}(x_{t,i})
 \nearrow P(c)
\ \ \text{as} \ \ n \to \infty,
\end{align}
where $\bar x_{t} = (x_{1},...,x_{t})$, $\Delta x_{t,i} = x_{t+1, i} - x_{t,i}$ and $\Delta x_{t} = x_{t+1} - x_{t}$. Denote $\phi_{t,n}^\oplus (x_t) =\sum_{i=1}^d  \phi_{t,i,n} (x_{t,i})$,  $h_{t,n}(\bar x_t) = \big{(} h_{t,1,n}(\bar x_t),..., h_{t,d,n}(\bar x_t)\big{)}$. Define
\begin{align}
\label{chidefinition}
\chi_{t,n}(x_{t}):= \sup_{x_1,...,x_{t-1}} \sum_{s=1}^{t-1} \big( \phi_{s,n}^\oplus (x_s) + h_{s,n}(\bar x_s) \cdot \Delta x_s \big)
\end{align}
with the convention $\chi_{1,n} = \chi_{N+1, n} \equiv 0$. Notice $\chi_{t,n}$ is a convex function on $\R^d$, since it is a supremum of affine functions of $x_t$. We now show
\begin{align}\label{chiineq}
\chi_{t,n} \le \chi_{t+1,n} - \phi_{t,n}^\oplus \ \text{ for all } t \in [N] \text{ and } n \in \N.
\end{align}
This inequality can be shown as follows:
\begin{align*}
\chi_{t+1,n}(x_{t+1}) &= \sup_{x_1,...,x_{t}} \sum_{s=1}^{t} \big( \phi_{s,n}^\oplus (x_s) + h_{s,n}(\bar x_s) \cdot \Delta x_s \big) \\
& \ge \sup_{\substack{x_1,...,x_{t-1} \\ x_{t} = x_{t+1}}} \sum_{s=1}^{t-1} \big( \phi_{s,n}^\oplus (x_s) + h_{s,n}(\bar x_s) \cdot \Delta x_s \big) +  \phi_{t,n}^\oplus (x_{t+1})\\
&= \chi_{t,n}(x_{t+1}) + \phi_{t,n}^\oplus (x_{t+1}).
\end{align*}
We can now establish the following crucial bound
\begin{align}\label{intbypart}
\int \chi_{t, n} \,d(\mu_t^\otimes - \mu_{t-1}^\otimes) \le C \ \text{ for all } t=2,...,N \text{ and } n \in \N.
\end{align}
To see this, by repeated application of \eqref{chiineq}, we have
\begin{align*}
\mu_t^\otimes(\chi_{t,n}) &\le \mu_t^\otimes(\chi_{t+1,n}) - \mu_t^\otimes(\phi_{t,n}^\oplus) \\
&\le \mu_{t+1}^\otimes(\chi_{t+1,n}) - \mu_t^\otimes(\phi_{t,n}^\oplus) \\
&\le \mu_{t+1}^\otimes(\chi_{t+2,n}) - \mu_{t+1}^\otimes (\phi_{t+1,n}^\oplus) - \mu_t^\otimes(\phi_{t,n}^\oplus) \\
&\le \dots \le - \sum_{s=t}^N \mu_s^\otimes (\phi_{s,n}^\oplus), 
\end{align*}
where the second inequality is due to $\mu_t^\otimes \preceq_c \mu_{t+1}^\otimes$ and convexity of $\chi$. Similarly, 
\begin{align*}
\mu_{t-1}^\otimes(\chi_{t,n}) &\ge \mu_{t-1}^\otimes(\chi_{t-1,n}) + \mu_{t-1}^\otimes(\phi_{t-1,n}^\oplus) \\
&\ge \mu_{t-2}^\otimes(\chi_{t-1,n}) + \mu_{t-1}^\otimes(\phi_{t-1,n}^\oplus) \\
&\ge \mu_{t-2}^\otimes(\chi_{t-2,n}) + \mu_{t-2}^\otimes(\phi_{t-2,n}^\oplus) + \mu_{t-1}^\otimes(\phi_{t-1,n}^\oplus) \\
&\ge \dots \ge  \sum_{s=1}^{t-1} \mu_s^\otimes (\phi_{s,n}^\oplus).
\end{align*}
The two inequalities combine to give
\begin{align}\label{chiineq2}
\int \chi_{t, n} \,d(\mu_t^\otimes - \mu_{t-1}^\otimes) \le  - \sum_{s=1}^{N} \mu_s^\otimes (\phi_{s,n}^\oplus) = - \mu (\phi_n)
\end{align}
which, in conjunction with \eqref{maximizing},  yields \eqref{intbypart}.

From \eqref{intbypart}, we can obtain local uniform boundedness of $\{\chi_{t,n}\}_n$ via Proposition \ref{conv}. We need to meet the proposition's second condition. For this, fix any $a \in I_1$, and let $L_{2,n}$ be an affine functions satisfying $L_{2,n} \le \chi_{2,n}$ and $L_{2,n}(a) = \chi_{2,n}(a)$. By linearity, $L_{2,n}(x_2) = \nabla L_{2,n}(x_1) \cdot (x_2 - x_1) + L_{2,n}(x_1)$, which allows us to modify \eqref{dual} by replacing $\phi_{1,n}^\oplus(x_1)$ with $\phi_{1,n}^\oplus(x_1)- L_{2,n}(x_1)$, $\phi_{2,n}^\oplus(x_2)$ with $\phi_{2,n}^\oplus(x_2) + L_{2,n}(x_2)$, and $h_{1,n}(x_1)$ with $h_{1,n}(x_1) - \nabla L_{2,n}(x_1)$ ($\nabla L_{2,n}$ is constant and does not depend on $x_1$). Notice this yields $\chi_{2,n} \ge 0$  and $\chi_{2,n}(a) =0$. We can continue subtracting appropriate linear functions $L_{t,n}$, $t=2,...,N$, and obtain
\begin{align}\label{convex0}
\chi_{t,n} \ge 0 \ \text{ and } \ \chi_{t,n}(a) =0 \quad \text{for all } n \in \N \text{ and } t \in [N].
\end{align}
Note that the modifications have no effect on the value $\mu(\phi_n)$.

To quantitatively state the local bound of $\chi$, let $\{\epsilon_k\}_k$ be a positive decreasing sequence tending to zero as $k \to \infty$, and write $I_{t,i} = ]a_{t,i}, b_{t,i}[$ where $-\infty \le a_{t,i} < b_{t,i} \le +\infty$. Then we define the compact interval $J_{t,i,k}:=[c_{t,i,k}, d_{t,i,k}]$ for $t \in [N-1]$, $i \in [d]$ and $k \in \N$ as follows:
\begin{align*}
&\text{If } a_{t,i} > -\infty, \text{ then define $c_{t,i,k }$ by:}\\
&\ \mu_{t+1,i} (a_{t,i})=0 \Rightarrow c_{t,i,k }:= a_{t,i} + \epsilon_k; \ \ \mu_{t+1,i} (a_{t,i}) > 0 \Rightarrow c_{t,i,k} := a_{t,i}; \\
&\text{If } b_{t,i} < +\infty, \text{ then define $d_{t,i,k }$ by:}\\
&\ \mu_{t+1,i} (b_{t,i})=0 \Rightarrow d_{t,i,k }:= b_{t,i} -  \epsilon_k; \ \ \mu_{t+1,i} (b_{t,i})>0 \Rightarrow d_{t,i,k }:= b_{t,i};  \\
&\text{If }  a_{t,i} = -\infty, \text{ then } c_{t,i,k} := -1/ \epsilon_k; \\
&\text{If }  b_{t,i} = +\infty, \text{ then }  d_{t,i,k} := +1/\epsilon_k.
\end{align*}
Set $J_{0,i,k} := J_{1,i,k} $. For example, if $\mu_{t+1,i} (a_{t,i}) = 0$ and $\mu_{t+1,i} (b_{t,i}) > 0$, then $J_{t,i,k} = [a_{t,i} + \epsilon_k, b_{t,i}]$. Let $\epsilon_1$ be so small so that $\mu_{t,i} (J_{t,i,1})>0$, $\mu_{t+1,i} (J_{t,i,1}) > 0$ for every $t,i$. Observe that $J_{t,i,k} \nearrow  J_{t,i}$ as $k \to \infty$. Let $J_{t,k} := J_{t,1,k} \times J_{t,2,k} \times ... \times J_{t,d,k}$. 
Then by Proposition \ref{conv}, we deduce there exists  $M_{k} \ge 0$ for each $k \in \N$ such that 
\begin{align}\label{boundconvex}
0 \le \sup_n \chi_{t,n} \le M_{k} \ \ \text{in} \ \, J_{t-1,k}.
\end{align}

{\bf Step 2.} Given an approximating dual maximizer $(\phi_n, h_n)_{n \in \N}$ (where we recall that each of $\phi_n = (\phi_{t,i,n})_{t,i}$ and $h_n = (h_{t,i,n})_{t,i}$ is in $\R^{Nd}$), our goal is to show pointwise convergence of $\phi_{t,i,n}$ to some function $\phi_{t,i}$ $\mu_{t,i}$-a.s. as $n \to \infty$, where $\phi_{t,i} \in \R \cup \{-\infty\}$ is $\mu_{t,i}$-a.s. finite.
But as noted in \cite{Lim23}, there is an obstacle for the convergence when $d \ge 2$, that is, in the duality relation \eqref{dual} one can always replace $(\phi_{t,i,n})_{t,i}$ by $(\phi_{t,i,n}+ C_{t,i,n})_{t,i} $ for any constants $C_{t,i,n} $ satisfying 
  $\sum_{t,i} C_{t,i,n} = 0$. This implies that the convergence cannot hold for any approximating dual maximizer. In view of this, our goal is to show that there exists a suitable approximating dual maximizer that yields the almost sure convergence. 

To this end, take any approximating dual maximizer $(\phi_n, h_n)_{n \in \N}$ satisfying \eqref{dual}, \eqref{maximizing}. By repeatedly applying \eqref{chiineq}, we deduce
\begin{align*}
C &\ge  -\sum_{s=1}^{N} \mu_s^\otimes (\phi_{s,n}^\oplus)\\
&\ge  \mu_N^\otimes (\chi_{N,n}) -\sum_{s=1}^{N-1} \mu_s^\otimes (\phi_{s,n}^\oplus) \\
&\ge  \mu_{N-1}^\otimes (\chi_{N,n}) -\sum_{s=1}^{N-1} \mu_s^\otimes (\phi_{s,n}^\oplus) \\
&=\,  \mu_{N-1}^\otimes (\chi_{N,n}) -  \mu_{N-1}^\otimes (\chi_{N-1,n}) +  \mu_{N-1}^\otimes (\chi_{N-1,n}) -\sum_{s=1}^{N-1} \mu_s^\otimes (\phi_{s,n}^\oplus) \\
&=\,  \parallel \chi_{N,n} -\chi_{N-1,n} - \phi_{N-1,n}^\oplus \parallel_{L^1( \mu_{N-1}^\otimes )}  +  \mu_{N-1}^\otimes (\chi_{N-1,n}) -\sum_{s=1}^{N-2} \mu_s^\otimes (\phi_{s,n}^\oplus) \\
&\ge\,  \parallel \chi_{N,n} -\chi_{N-1,n} - \phi_{N-1,n}^\oplus \parallel_{L^1( \mu_{N-1}^\otimes )}  +  \mu_{N-2}^\otimes (\chi_{N-1,n}) -\sum_{s=1}^{N-2} \mu_s^\otimes (\phi_{s,n}^\oplus) \\
&\ge \dots \ge  \sum_{s=2}^N \parallel \chi_{s,n} -\chi_{s-1,n} - \phi_{s-1,n}^\oplus \parallel_{L^1( \mu_{s-1}^\otimes )}
\end{align*}
where the third and sixth inequality is due to the  convexity of $\chi$ with $\mu_t^\otimes \preceq_c \mu_{t+1}^\otimes$, and the fifth equality is by the nonnegativity $\chi_{t,n} -\chi_{t-1,n} - \phi_{t-1,n}^\oplus \ge 0$ from \eqref{chiineq}. On the other hand, the nonpositivity \eqref{dual} yields
\begin{align*}
\sum_{s=1}^{N} \big( \phi_{s,n}^\oplus (x_s) + h_{s,n}(\bar x_s) \cdot \Delta x_s \big) \le \chi_{N,n} (x_N) + \phi^\oplus_{N,n} (x_N) \le 0
\end{align*}
where the first inequality follows by taking supremum over $x_1,...,x_{N-1}$ (recall \eqref{chidefinition} and $h_{N,n} \equiv 0$). Integrating with any $\pi \in {\rm VMT}(\mu)$ yields $\parallel \chi_{N,n} (x_N) + \phi^\oplus_{N,n}\parallel_{L^1( \mu_{N}^\otimes )} \le -\sum_{s=1}^{N} \mu_s^\otimes (\phi_{s,n}^\oplus) \le C$. We thus conclude
\be\label{L1bound}
\parallel \chi_{t+1,n} -\chi_{t,n} -
 \phi_{t,n}^\oplus \parallel_{L^1( \mu_{t}^\otimes )}\, \le C \ \text{ for all } n \in \N, t \in [N].
\ee
This uniform $L^1$ bound, in conjunction with the local uniform bound \eqref{boundconvex} and Koml{\'o}s compactness theorem, can imply the desired almost sure convergence that we now present. For this, we will extend the argument given in \cite{Lim23} for two-period case $N=2$ into arbitrary $N \ge 2$. For each $k \in \N$, let $\mu_{t,i,k}$ be the restriction of $\mu_{t,i}$ on $J_{t-1,i,k}$ (where $J_{0,i,k} := J_{1,i,k}$) then normalized to be a probability distribution. Let $\mu^\otimes_{t,k} = \otimes_i \mu_{t,i,k}$, so that $\mu^\otimes_{t,k}(J_{t,k})=1$ where $J_{t,k} := \otimes_i J_{t,i,k}$. Define
\begin{align*}
    v_{t,i,k,n} := \int \phi_{t,i,n} \,d\mu_{t,i,k}, \quad t \in [N], i \in [d], k\in \N, n \in \N.
\end{align*}
For each $k \in\N$, we will show that there exists $C= C(k)$ such that
\begin{align}\label{supbound}
\sup_n \parallel \phi_{t,i,n} - v_{t,i,k,n} \parallel_{L^1(\mu_{t,i,k})} \, \le C.
\end{align}
To see this, observe that \eqref{boundconvex}, \eqref{L1bound} and the fact $J_{t-1,k} \subset J_{t,k}$ imply
\begin{align*}
&\parallel \phi_{t,n}^\oplus \parallel_{L^1(\mu^\otimes_{t,k})} \\ 
&\le\,  \parallel \chi_{t+1,n} -\chi_{t,n} -
 \phi_{t,n}^\oplus \parallel_{L^1(\mu^\otimes_{t,k})} + \parallel \chi_{t,n} - \chi_{t+1,n}  \parallel_{L^1(\mu^\otimes_{t,k})} \\
&\le C + M_k =: C
\end{align*}
where we used $\chi_{t+1,n} -\chi_{t,n} -
 \phi_{t,n}^\oplus \ge 0$ to get the bound $ \parallel \chi_{t+1,n} -\chi_{t,n} -
 \phi_{t,n}^\oplus \parallel_{L^1(\mu^\otimes_{t,k})}\, \le C$ from \eqref{L1bound}. From this, we obtain the bound
\begin{align}\label{boundv}
\bigg| \sum_{i=1}^d v_{t,i,k,n} \bigg| \le \, \parallel \phi_{t,n}^\oplus \parallel_{L^1(\mu^\otimes_{t,k})}  \, \le C \, \text{ for all } n,
\end{align}
where the first inequality is by Jensen's inequality.
Next, because $\phi_{t,n}^\oplus \le M_k$ on $J_{t-1,k}$ by \eqref{chiineq} and \eqref{boundconvex}, by taking supremum, we have
\begin{align*}
\sum_{i=1}^d \sup_{x_{t,i} \in J_{t,i,k}} \phi_{t,i,n} (x_{t,i}) \le M_k  \, \text{ for all } n,
\end{align*}
and note that clearly $v_{t,i,k,n} \le \sup_{x_{t,i} \in J_{t,i,k}} \phi_{t,i,n} (x_{t,i})$, so in particular,
\begin{align*}
\sup_{x_{t,1} \in J_{t,1,k}} \phi_{t,1,n} (x_{t,1})+ \sum_{i=2}^d v_{t,i,k,n}  \le M_k.
\end{align*}
Define  $\hat v_{t,1,k,n} := -\sum_{i=2}^d v_{t,i,k,n}$. Since $\phi_{t,n}^\oplus  \le M_{k}$ on $J_{t-1,k}$, we have
\begin{align*}
C \ge \,\,\parallel M_{k} - \phi_{t,n}^\oplus \parallel_{L^1(\mu^\otimes_{t,k})}\, = M_k - \int ( \phi_{t,1,n} + \sum_{i=2}^d v_{t,i,k,n}) d\mu_{t,1,k}.
\end{align*}
This implies that $ \sup_n \parallel  \phi_{t,1,n} - \hat v_{t,1,k,n} \parallel_{L^1(\mu_{t,1,k})}$ is bounded, and then by \eqref{boundv}, $ \sup_n \parallel  \phi_{t,1,n} - v_{t,1,k,n} \parallel_{L^1(\mu_{t,1,k})}$ is bounded. This yields \eqref{supbound}.

We can now apply the Koml{\'o}s lemma, which states that every $L^1$-bounded sequence of real functions contains a subsequence such that the arithmetic means of all its subsequences converge pointwise almost everywhere. 

Let $\tilde v_{t,1,k,n}=  \frac{1}{n} \sum_{m=1}^n \hat v_{t,1,k,m}$, and $\tilde v_{t,i,k,n}=  \frac{1}{n} \sum_{m=1}^n v_{t,i,k,m}$ for $i \ge 2$.
Now for each $k \in \N$, a repeated application of Koml{\'o}s lemma yields that there exists a subsequence $\{\phi_{t,i,k,n}\}_n$ of $\{\phi_{t,i,n}\}_n$, such that

(i) $\{\phi_{t,i,k+1,n}\}_n$ is a further subsequence of $\{\phi_{t,i,k,n}\}_n$, and

(ii) $\tilde \phi_{t,i,k,n}(x_{t,i}) - \tilde v_{t,i,k,n}$ converges $\mu_{t,i,k}$ - a.s. \, as \, $n \to \infty$,\\
where $\tilde \phi_{t,i,k,n}:= \frac{1}{n} \sum_{m=1}^n \phi_{t,i,k,m}$. Note that for each $k$, our choice of a subsequence index can be made identical for every $t$ and $i$,  since there are finitely many indices of $t,i$.  Then we select the diagonal sequence  
\[
\Phi_{t,i,n} := \phi_{t,i,n,n}\]
 and again define 
 \begin{align*}
 w_{t,i,k,n} &= \int \Phi_{t,i,n} d\mu_{t,i,k}, \q 2\le i \le d,\\
 \hat w_{t,1,k,n} &= -\sum_{i=2}^d w_{t,i,k,n}, \\
 \tilde \Phi_{t,i,n}(x_{t,i}) &= \frac{1}{n} \sum_{m=1}^n \Phi_{t,i,m} (x_{t,i}), \\
 \tilde w_{t,1,k,n} &=  \frac{1}{n} \sum_{m=1}^n \hat w_{t,1,k,m}, \\
 \tilde w_{t,i,k,n} &=  \frac{1}{n} \sum_{m=1}^n w_{t,i,k,m}, \q 2 \le i \le d.
 \end{align*}
 We finally claim that 
 \begin{align}\label{convergePhi}
\tilde \Phi_{t,i,n}(x_{t,i})  - \tilde w_{t,i,1,n} \, \text{converges } \, \mu_{t,i} - a.s. \, \text{ for all } t \in [N], i \in [d].
\end{align}
Note that the dependence on $k$ has now been removed. To prove the claim, since $\{\Phi_{t,i,n}\}_n$ is a subsequence of $\{\phi_{t,i,k,n}\}_n$ for every $k\in \N$,   Koml{\'o}s lemma implies
 \begin{align}\label{convergePhik}
\tilde \Phi_{t,i,n}(x_{t,i})  - \tilde w_{t,i,k,n} \ \text{converges } \, \mu_{t,i,k} - a.s. \, \text{ for all } t \text{ and } i. 
\end{align}
In particular, both $\{\tilde \Phi_{t,i,n}(x_{t,i})  - \tilde w_{t,i,1,n}\}_n$ and $\{\tilde \Phi_{t,i,n}(x_{t,i})  - \tilde w_{t,i,k,n}\}_n$ converge $ \mu_{t,i,1}$ - a.s. as $n\to \infty$, hence their difference $ \{\tilde w_{t,i,1,n} -  \tilde w_{t,i,k,n}\}_n$ also converges for any fixed $k$. With \eqref{convergePhik}, this implies \eqref{convergePhi}. Finally, the fact that $\sum_{i=1}^d \tilde w_{t,i,1,n} = 0$ allows us to replace the approximating dual maximizer $(\phi_n, h_n)_{n}$ by $(\psi_n, \tilde h_n)_{n}$, where $\psi_{t,i,n} :=  \tilde \Phi_{t,i,n}(x_{t,i})  - \tilde w_{t,i,1,n}$ and $\tilde h_{t,i,n}$ is a suitable Ces{\`a}ro mean of a subsequence of $(h_{t,i,n})_n$ which is chosen consistently with the selection of $\tilde \Phi_{t,i,n}$.
\\

{\bf Step 3.} We will prove the convergence of $\{\chi_{t,n}\}_n$ defined in \eqref{chidefinition}. Let $(\phi_n, h_n)_{n}$ be an approximating dual maximizer such that 
\be\label{asconverge}
\{ \phi_{t,i,n} \}_n \text{ converges to a function } \phi_{t,i} \ \mu_{t,i}-a.s. \text{ as } n \to \infty. 
\ee
Now it is unclear that those nice properties $\chi_{t,n}$ enjoyed in Step 2, e.g., \eqref{convex0} or \eqref{boundconvex}, keep holding, because $(\phi_n, h_n)_{n}$ that we now consider is a Ces{\`a}ro mean of a subsequence of an approximating dual maximizer. But notice that this implies at least the upper bound in \eqref{boundconvex} continues to hold, that is,
\begin{align}\label{boundconvex1}
\sup_n \chi_{t,n} \le M_{k} \ \ \text{on} \ \ J_{t-1,k}.
\end{align}
We will show that the normalization of $\chi_{t,n}$ \eqref{convex0} can be restored. Fix any $a \in I_1$ as in Step 2, and define $\phi^\oplus_t := \sum_{i} \phi_{t,i}$. First, \eqref{asconverge}  implies that there exists $a_t \in I_t$ for every $t \in [N]$ such that 
\be
\lim_{n \to \infty} \phi_{t,n}^\oplus (a_t) = \phi_t^\oplus (a_t) \in \R.
\ee
In view of \eqref{chiineq} which gives $\phi^\oplus_{1,n} \le \chi_{2,n}$, this implies
\be\label{goodlowbound}
\inf_n \chi_{2,n} (a_1) > -\infty.
\ee
On the other hand, since $I_1 = {\rm int} (J_1)$ and $J_{1,k} \nearrow J_{1}$, for large enough $k$ we have $\{a, a_1\} \subset {\rm int}(J_{1,k})$. Now \eqref{boundconvex1}, \eqref{goodlowbound} imply that both $\chi_{2,n}(a)$ and $\nabla \chi_{2,n} (a)$ are uniformly bounded in $n$, where $\nabla \chi_{2,n} (a) \in \partial \chi_{2,n} (a)$ is  a subgradient  of the convex function $\chi_{2,n}$ at $a$. Hence by taking a subsequence, we can assume that $\{\chi_{2,n}(a)\}_n$ and $\{\nabla \chi_{2,n} (a)\}_n$ both converge. Then as in Step 1, define an affine function $L_{2,n}(y) = \chi_{2,n}(a) + \nabla \chi_{2,n} (a) \cdot (y-a)$, and replace $\phi_{1,n}^\oplus(x_1)$ with $\phi_{1,n}^\oplus(x_1) - L_{2,n}(x_1)$, $\phi_{2,n}^\oplus(x_2)$ with $\phi_{2,n}^\oplus(x_2) + L_{2,n}(x_2)$, and finally $h_{1,n}(x_1)$ with $h_{1,n}(x_1) - \nabla \chi_{2,n}(a)$. This yields $\chi_{2,n}(a) = \nabla \chi_{2,n}(a) = 0$ for all $n$, while the a.s. convergence of $\phi_{t,i, n}$ and the bound \eqref{boundconvex1} are retained. Next, the inequality $\chi_{3,n} \ge \chi_{2,n} + \phi^\oplus_{2,n}$ with $\chi_{2,n} \ge 0$ yields $\inf_n \chi_{3,n}(a_2) > -\infty$. Thus we can repeat the argument and achieve the normalization \eqref{convex0}, while the a.s. convergence of $\phi_{t,i, n}$ and the bound \eqref{boundconvex1} are still retained. Now  \eqref{intbypart} and Proposition \ref{conv} yield $\lim_{n \to \infty} \chi_{t,n} = \chi_t$ pointwise on $J_{t-1}$.
\\

{\bf Step 4.}  We have obtained the almost sure limit functions $(\phi_{t,i})_{t,i}$. We may define $\phi_{t,i} := -\infty$ on a $\mu_{t,i}$-null set which includes $\R \setminus I_t$, so that they are defined everywhere on $\R$. We will now show there exist functions $h_t = (h_{t,i})_{i} : \R^{td} \to \R^d$ for all $t \in [N]$ with $h_N \equiv 0$, such that 
\begin{align}\label{duallimit}
\sum_{t=1}^{N} \big( \phi_{t}^\oplus (x_t) + h_{t}(\bar x_t) \cdot \Delta x_t \big) \le c(x) \ \ \text{for all } x \in \R^{Nd}.
\end{align}
For any function $f : \R^d \to \R \cup \{+\infty\}$ which is bounded below by an affine function, let ${\rm conv}[f]:\R^d \to \R \cup \{+\infty\}$ denote the lower semi-continuous convex envelope of $f$, that is the supremum of all affine functions $l$ satisfying $ l \le f$ (If there is no such $l$, let ${\rm conv}[f] \equiv -\infty$.) We will inductively obtain $h_{N-1}, h_{N-2},...,h_1$. Let us rewrite \eqref{dual} as
\begin{align}\label{a1}
\sum_{t=1}^{N-1} \big( \phi_{t,n}^\oplus (x_t) + h_{t,n}(\bar x_t) \cdot \Delta x_t \big) \le c(x) -  \phi_{N,n}^\oplus (x_N).
\end{align}
Define $H_{N-1,n}(\bar x_{N-1},x_N) = {\rm conv}[c(\bar x_{N-1},\,\cdot\,) - \phi_{N,n}^\oplus (\,\cdot\,)](x_N)$. We  have
\begin{align*}
\sum_{t=1}^{N-1} \big( \phi_{t,n}^\oplus (x_t) + h_{t,n}(\bar x_t) \cdot \Delta x_t \big) \le H_{N-1,n}(\bar x_{N-1},x_N) \le c(x) -  \phi_{N,n}^\oplus (x_N)
\end{align*}
because the left hand side is affine in $x_N$. If we let $x_N = x_{N-1}$, we get
\begin{align}\label{a2}
\sum_{t=1}^{N-2} \big( \phi_{t,n}^\oplus (x_t) + h_{t,n}(\bar x_t) \cdot \Delta x_t \big) \le H_{N-1,n}(\bar x_{N-1},x_{N-1}) -  \phi_{N-1,n}^\oplus (x_{N-1}).
\end{align}
Notice \eqref{a1} and \eqref{a2} have the same structure. With the convention $H_{N,n} (x) := c(x)$,
 this allows us to inductively deduce, backward in $t$,
 \begin{align}\label{a3}
\sum_{s=1}^{t} \big( \phi_{s,n}^\oplus (x_s) + h_{s,n}(\bar x_s) \cdot \Delta x_s \big) \le H_{t+1,n}(\bar x_{t+1},x_{t+1}) -  \phi_{t+1,n}^\oplus (x_{t+1})
\end{align}
for $t=1,...,N-1$, where
\begin{align}\label{a4}
H_{t,n}(\bar x_{t},x_{t+1}) := {\rm conv}[H_{t+1,n}(\bar x_{t},\,\cdot\,) - \phi_{t+1,n}^\oplus (\,\cdot\,)](x_{t+1})
\end{align}
with an abuse of notation $H_{t+1,n}(\bar x_{t},x_{t+1}) := H_{t+1,n}(\bar x_{t+1},x_{t+1}) $. 

Now by dropping $n$, we analogously define
\begin{align}\label{a5}
H_{t}(\bar x_{t},x_{t+1}) := {\rm conv}[H_{t+1}(\bar x_{t},\,\cdot\,) - \phi_{t+1}^\oplus (\,\cdot\,)](x_{t+1})
\end{align}
with $H_{N} (x) := c(x)$. Next, since the $\limsup$ of convex functions is convex, in conjunction with the almost sure convergence, we have
\begin{align*}
\limsup_{n \to \infty} H_{N-1,n}(\bar x_{N-1},x_N) 
&= \limsup_{n \to \infty} {\rm conv}[c(\bar x_{N-1},\,\cdot\,) - \phi_{N,n}^\oplus (\,\cdot\,)](x_N) \\
&\le {\rm conv}[\limsup_{n \to \infty}\big{(} c(\bar x_{N-1},\,\cdot\,) - \phi_{N,n}^\oplus (\,\cdot\,) \big{)}](x_N)  \\
&\le {\rm conv}[c(\bar x_{N-1},\,\cdot\,) - \phi_{N}^\oplus (\,\cdot\,)](x_N) \\
&= H_{N-1}(\bar x_{N-1},x_N).
\end{align*}
This allows us to inductively deduce, for $t=1,...,N-1$,
\begin{align*}
\limsup_{n \to \infty} H_{t,n}(\bar x_{t},x_{t+1}) 
&= \limsup_{n \to \infty} {\rm conv}[H_{t+1,n}(\bar x_{t},\,\cdot\,) - \phi_{t+1,n}^\oplus (\,\cdot\,)](x_{t+1}) \\
&\le {\rm conv}[\limsup_{n \to \infty}\big{(} H_{t+1,n}(\bar x_{t},\,\cdot\,) - \phi_{t+1,n}^\oplus (\,\cdot\,)\big{)}](x_{t+1})  \\
&\le {\rm conv}[H_{t+1}(\bar x_{t},\,\cdot\,) - \phi_{t+1}^\oplus (\,\cdot\,)](x_{t+1}) \\
&= H_{t}(\bar x_{t},x_{t+1}).
\end{align*}
We need to discuss continuity of the convex function $x_{t+1} \mapsto H_{t}(\bar x_{t}, x_{t+1})$. The following inequality from \eqref{a5}
\be
H_{t}(\bar x_{t},x_{t+1}) \le H_{t+1}(\bar x_{t+1}, x_{t+1}) - \phi_{t+1}^\oplus (x_{t+1}) \nn 
\ee 
becomes $H_{N-1}(\bar x_{N-1},x_{N}) \le c(x) - \phi_{N}^\oplus (x_{N})$ when $t=N-1$. Then the $\mu_N^\otimes$-a.s. finiteness of $\phi_{N}^\oplus$ gives, by convexity, $H_{N-1}(\bar x_{N-1},x_{N}) < \infty$ if $x_N \in J_{N-1}$.  Backward induction in $t$ then gives $H_{t}(\bar x_{t},x_{t+1}) < \infty$ if $x_{t+1} \in J_{t}$. This implies that for any $\bar x_t \in \R^{td}$, if there exists $y_0 \in \R^d$ such that $H_{t}(\bar x_{t}, y_0) > -\infty$, then $y \mapsto H_{t}(\bar x_{t}, y)$ is real-valued thus continuous in $J_{t}$. Now \eqref{a3},  \eqref{a4} gives
\begin{align*}
\phi_{1,n}^\oplus (x_1) + h_{1,n}(x_1) \cdot \Delta x_1 
 \le H_{1,n}( x_{1},x_{2}) \le H_{2,n}( x_{1},x_{2}, x_{2}) -  \phi_{2,n}^\oplus (x_{2}).
\end{align*}
Letting $x_2=x_1$ gives $\phi_{1,n}^\oplus (x_1) \le H_{1,n}( x_{1},x_{1})$. From this and the almost sure convergence, taking $\limsup$ yields
\begin{align*}
\phi_{1}^\oplus (x_1) \le H_{1}( x_{1},x_{1}) \ \text{ and } \ H_{1}( x_{1},x_{2}) \le H_{2}( x_{1},x_{2},x_{2}) -  \phi_{2}^\oplus (x_{2}). 
\end{align*}
Set $A_t:= \{ x_t \in \R^d \, | \, \phi_{t}^\oplus (x_t) \in \R \}$, $t \in [N]$, and note that $A_t \subset I_t$. Since $x_2 \mapsto H_1(x_1,x_2)$ is continuous in $J_1$ for every $x_1 \in A_1$, the subdifferential $\partial H_1(x_1, \,\cdot\,) (x_2)$ is nonempty, convex and compact for every $x_2 \in I_1 = {\rm int}(J_1)$. This allows us to choose a measurable function $h_1 : A_1 \to \R^d$ satisfying $h_1(x_1) \in \partial H_1(x_1, \,\cdot\,)(x_1)$. 
Then for $x_1 \in A_1$, we have
\begin{align*}
\phi_{1}^\oplus (x_1) + h_1(x_1) \cdot (x_2-x_1) &\le H_1(x_1,x_1) + h_1(x_1) \cdot (x_2-x_1) \\
 &\le H_1(x_1,x_2) \\
 &\le H_{2}( x_{1},x_{2}, x_{2}) -  \phi_{2}^\oplus (x_{2}). 
 \end{align*}
In particular, for $x_1 \in A_1$ and $x_2 \in A_2$, it holds $H_{2}( x_{1},x_{2}, x_{2}) > -\infty$. Hence again we can choose $h_2 : A_1 \times A_2 \to \R^d$ that satisfies $h_2(x_1,x_2) \in \partial H_2 (x_1, x_2, \,\cdot\,)(x_2)$. Then for every $x_1 \in A_1$ and $x_2 \in A_2$, we have
\begin{align*}
&\phi_{1}^\oplus (x_1) +  \phi_{2}^\oplus (x_{2}) + h_1(x_1) \cdot (x_2-x_1) + h_2(x_1,x_2) \cdot (x_3-x_2) \\&\le H_{2}( x_{1},x_{2}, x_{2}) + h_2(x_1,x_2) \cdot (x_3-x_2)\\
&\le H_{2}( x_{1},x_{2}, x_{3})\\
&\le H_{3}( x_{1},x_{2}, x_{3}, x_{3}) - \phi_{3}^\oplus (x_{3}).
 \end{align*}
By induction in $t$, we obtain $h_t : A_1 \times \dots \times A_t \to \R^d$ (with $h_N \equiv 0$), satisfying \eqref{duallimit} as desired. We may define $h_t = 0$ in $\R^{td} \setminus A_1 \times \dots \times A_t$, noting that the left hand side of \eqref{duallimit} is $-\infty$ if $x_t \notin A_t$ for some $t$.
\\

{\bf Step 5.} We will show that for any functions $h_t : \R^{td} \to \R^d$, $t=1,...,N$ with $h_N \equiv 0$ satisfying  \eqref{duallimit} (whose existence was shown in Step 4), and for any minimizer $\pi^*$ for the problem \eqref{VMOT}, it holds
\begin{align}\label{pointwisedualeq}
\sum_{t=1}^{N} \big( \phi_{t}^\oplus (x_t) + h_{t}(\bar x_t) \cdot \Delta x_t \big) = c(x), \quad \pi^* - a.s..
\end{align}
In other words, every minimizer $\pi^*$ is concentrated on the contact set
\[
\Gamma := \bigg\{x \in \R^{Nd} \,\bigg|\, \sum_{t=1}^{N} \big( \phi_{t}^\oplus (x_t) + h_{t}(\bar x_t) \cdot \Delta x_t \big) = c(x) \bigg\}
\]
whenever $\{h_t\}_t$ is chosen to satisfy \eqref{duallimit}. This will complete the proof.

Recall $\phi_{t,i,n} \to \phi_{t,i}$ $\mu_{t,i}$-a.s. and $\chi_{t,n} \to \chi_t$ in $J_{t-1}$ where $\chi_{t,n}$ is defined in \eqref{chidefinition} with $\chi_t$ being its limit. For any $\pi \in {\rm VMT}(\mu)$ (not necessarily an optimizer), we have $c \in L^1(\pi)$ by the assumption of Theorem \ref{main}. We claim:
\begin{align}\label{claim1}
\limsup_{n \to \infty}& \int \sum_{t=1}^{N} \big( \phi_{t,n}^\oplus (x_t) + h_{t,n}(\bar x_t) \cdot \Delta x_t \big)  d\pi  \\
&\le \int \sum_{t=1}^{N} \big( \phi_{t}^\oplus (x_t) + h_{t}(\bar x_t) \cdot \Delta x_t \big) d\pi. \nn
\end{align}
To see how the claim implies \eqref{pointwisedualeq}, let $\pi^*$ be any minimizer for \eqref{VMOT} (which exists by the assumption on $c$). Then 
$P(c) = \int c \,d\pi^*$, hence
\begin{align*}
P(c)&=\lim_{n \to \infty}  \int \sum_{t=1}^{N} \big( \phi_{t,n}^\oplus (x_t) + h_{t,n}(\bar x_t) \cdot \Delta x_t \big) d\pi^* \\ 
&\le \int \sum_{t=1}^{N} \big( \phi_{t}^\oplus (x_t) + h_{t}(\bar x_t) \cdot \Delta x_t \big) d\pi^* \\ 
&\le  \int c(x) \,d\pi^* = P(c)
\end{align*}
hence equality holds throughout. Notice that this implies \eqref{pointwisedualeq}.

To prove \eqref{claim1}, we will extend the argument given in \cite{bnt}, \cite{Lim23} into the current multi-period vector-valued setting. Fix any $\pi \in {\rm VMT}(\mu)$. Denoting $\pi = {\rm Law}(X)$ where $X=(X_1,...,X_N)$ is an $\R^d$-valued martingale under $\pi$, we let $\pi_t := {\rm Law}(X_t)$. Then as in Step 2 (but using $\pi_t \preceq_c \pi_{t+1}$ instead of $\mu^\otimes_t \preceq_c \mu^\otimes_{t+1}$), by \eqref{maximizing}, \eqref{chiineq}, we have (cf. \eqref{L1bound})
\be\label{L1bound2}
\parallel \chi_{t+1,n} -\chi_{t,n} -
 \phi_{t,n}^\oplus \parallel_{L^1( \pi_{t})} \, \le C \ \text{ for all } n \in \N, t \in [N].
\ee
From this, as $\phi^\oplus_{t,n} \to \phi^\oplus_{t}$ and $\chi_{t,n} \to \chi_t$, we deduce by Fatou's lemma,
\begin{align*}
&\chi_{t+1} -\chi_{t} -
 \phi_{t}^\oplus \in L^1(\pi_t), \ \text{ and}\\
\limsup_{n \to \infty} \int (&\phi_{t,n}^\oplus + \chi_{t,n} - \chi_{t+1,n}) \,d\pi_t \le  \int (\phi_{t}^\oplus + \chi_{t} - \chi_{t+1} )\,d\pi_t,
\end{align*}
recalling $\chi_{1,n} = \chi_{N+1,n} \equiv 0$ and $\pi_{t} (J_{t-1}) = 1$. This allows us to deduce
\begin{align}
\limsup_{n \to \infty}  &\int \sum_{t=1}^{N} \big( \phi_{t,n}^\oplus (x_t) + h_{t,n}(\bar x_t) \cdot \Delta x_t \big) d\pi \nn \\
= \limsup_{n \to \infty}  &\int \sum_{t=1}^{N} \big( \phi_{t,n}^\oplus (x_t) + \chi_{t,n}(x_t) - \chi_{t+1,n}(x_t) \nn \\
&\ \ \q\q\q - \chi_{t,n}(x_t) + \chi_{t+1,n}(x_t)+ h_{t,n}(\bar x_t) \cdot \Delta x_t \big) d\pi \nn \\
\le  \sum_{t=1}^N \int &\big( \phi_{t}^\oplus + \chi_{t} - \chi_{t+1} \big) d\pi_t  \label{a8} \\
+ \limsup_{n \to \infty} &\int \sum_{t=1}^{N-1} \big{(}  \chi_{t+1,n}(x_t)  - \chi_{t+1,n}(x_{t+1})  + h_{t,n}(\bar x_t) \cdot \Delta x_t  \big{)} \, d\pi, \nn
\end{align}
since $ \sum_{t=1}^N \chi_{t,n}(x_{t}) = \sum_{t=1}^N \chi_{t+1,n}(x_{t+1})$. Now denote $\bar X_t = (X_1,...,X_t)$, $\pi^t := {\rm Law}(\bar X_t) \in \cP(\R^{td})$. Then we can write $\pi^{t+1} = \pi_{\bar x_t} \otimes \pi^t$, where $\pi_{\bar x_t} \in \cP(\R^d)$ is the conditional distribution of $X_{t+1}$ given $\bar X_t = \bar x_t$ under the martingale law $\pi$. Martingale property means that $\int y \, d \pi_{\bar x_t}(y) = x_t$. For each $t$, choose a sequence of functions $\xi_{t,n} : I_t \to \R^d$ satisfying $\xi_{t,n}(x_t) \in \partial \chi_{t+1,n}(x_t)$. Then we can compute
\begin{align*}
&\int \sum_{t=1}^{N-1} \big{(}  \chi_{t+1,n}(x_t)  - \chi_{t+1,n}(x_{t+1})  + h_{t,n}(\bar x_t) \cdot \Delta x_t  \big{)} \, d\pi \\
&= \iint \sum_{t=1}^{N-1} \big{(}  \chi_{t+1,n}(x_t)  - \chi_{t+1,n}(x_{t+1})  + h_{t,n}(\bar x_t) \cdot \Delta x_t  \big{)} d\pi_{\bar x_{N-1}}(x_N) d\pi^{N-1}(\bar x_{N-1}) \\
&= \int \bigg( \int \big( \chi_{N,n}(x_{N-1})  - \chi_{N,n}(x_{N})  + h_{N-1,n}(\bar x_{N-1}) \cdot \Delta x_{N-1}\big)  d\pi_{\bar x_{N-1}}(x_N)\\
&\q\q\q + \sum_{t=1}^{N-2} \big{(}  \chi_{t+1,n}(x_t)  - \chi_{t+1,n}(x_{t+1})  + h_{t,n}(\bar x_t) \cdot \Delta x_t  \big{)}\bigg)  d\pi^{N-1}(\bar x_{N-1}) \\
&= \int \bigg( \int \big( \chi_{N,n}(x_{N-1})  - \chi_{N,n}(x_{N})  + \xi_{N-1,n}( x_{N-1}) \cdot \Delta x_{N-1}\big)  d\pi_{\bar x_{N-1}}(x_N)\\
&\q\q\q + \sum_{t=1}^{N-2} \big(\chi_{t+1,n}(x_t)  - \chi_{t+1,n}(x_{t+1})  + h_{t,n}(\bar x_t) \cdot \Delta x_t  \big{)} \bigg)  d\pi^{N-1}(\bar x_{N-1}),
\end{align*}
where the last equality is due to the martingale property
\begin{align}\label{a7}
&\int h_{N-1,n}(\bar x_{N-1}) \cdot \Delta x_{N-1} \, d\pi_{\bar x_{N-1}}(x_N) \\
&= \int \xi_{N-1,n}( x_{N-1}) \cdot \Delta x_{N-1} \, d\pi_{\bar x_{N-1}}(x_N) \nn \\
&=0. \nn
\end{align}
Furthermore, by definition of $\xi$, we have
\be
\chi_{N,n}(x_{N-1})  - \chi_{N,n}(x_{N})  + \xi_{N-1,n}( x_{N-1}) \cdot \Delta x_{N-1} \le 0. \nn
\ee
This allows us to further disintegrate $\pi^{N-1} = \pi_{\bar x_{N-2}} \otimes \pi^{N-2}$ and apply the same argument. By repeating, we obtain
\begin{align*}
&\int \sum_{t=1}^{N-1} \big{(}  \chi_{t+1,n}(x_t)  - \chi_{t+1,n}(x_{t+1})  + h_{t,n}(\bar x_t) \cdot \Delta x_t  \big{)} \, d\pi \\
&= \int\dots \int \big( \chi_{N,n}(x_{N-1})  - \chi_{N,n}(x_{N})  + \xi_{N-1,n}( x_{N-1}) \cdot \Delta x_{N-1}\big)  d\pi_{\bar x_{N-1}}(x_N)\\
&\ + \big( \chi_{N-1,n}(x_{N-2})  - \chi_{N-1,n}(x_{N-1})  + \xi_{N-2,n}( x_{N-2}) \cdot \Delta x_{N-2}\big)  d\pi_{\bar x_{N-2}}(x_{N-1})\\
&\ +\dots + \big( \chi_{2,n}(x_{1})  - \chi_{2,n}(x_{2})  + \xi_{1,n}( x_{1}) \cdot \Delta x_{1}\big)  d\pi_{\bar x_{1}}(x_{2}) d\pi^1(x_1).
\end{align*}
Since $\chi_{t+1,n}(x_{t})  - \chi_{t+1,n}(x_{t+1})  + \xi_{t,n}( x_{t}) \cdot \Delta x_{t} \le 0$ for all $t$, repeated application of Fatou's lemma allows $\limsup$ to continue to penetrate into the innermost integral. Along the way, we also use the inequality $\limsup (a_n + b_n) \le \limsup a_n + \limsup b_n$. This eventually yield
\begin{align*}
&\limsup_{n\to\infty} \int \sum_{t=1}^{N-1} \big{(}  \chi_{t+1,n}(x_t)  - \chi_{t+1,n}(x_{t+1})  + h_{t,n}(\bar x_t) \cdot \Delta x_t  \big{)} \, d\pi \\
&\le \int\dots \int \big( \chi_{N}(x_{N-1})  - \chi_{N}(x_{N})  + \xi_{N-1}( x_{N-1}) \cdot \Delta x_{N-1}\big)  d\pi_{\bar x_{N-1}}(x_N)\\
&\ + \big( \chi_{N-1}(x_{N-2})  - \chi_{N-1}(x_{N-1})  + \xi_{N-2}( x_{N-2}) \cdot \Delta x_{N-2}\big)  d\pi_{\bar x_{N-2}}(x_{N-1})\\
&\ +\dots + \big( \chi_{2}(x_{1})  - \chi_{2}(x_{2})  + \xi_{1}( x_{1}) \cdot \Delta x_{1}\big)  d\pi_{\bar x_{1}}(x_{2}) d\pi^1(x_1)
\end{align*}
for some $\xi_t(x_t)  \in \partial \chi_{t+1}(x_t)$ which is a limit point of the bounded sequence $\{ \xi_{t,n}(x_t)\}_n$. Lastly, the martingale property \eqref{a7} allows us to substitute $\xi_t(x_t)$ back to $h_t(\bar x_t)$ and get
\begin{align*}
&\limsup_{n\to\infty} \int \sum_{t=1}^{N-1} \big{(}  \chi_{t+1,n}(x_t)  - \chi_{t+1,n}(x_{t+1})  + h_{t,n}(\bar x_t) \cdot \Delta x_t  \big{)} \, d\pi \\
&\le \int\dots \int \big( \chi_{N}(x_{N-1})  - \chi_{N}(x_{N})  + h_{N-1}( \bar x_{N-1}) \cdot \Delta x_{N-1}\big)  d\pi_{\bar x_{N-1}}(x_N)\\
&\ + \big( \chi_{N-1}(x_{N-2})  - \chi_{N-1}(x_{N-1})  + h_{N-2}( \bar x_{N-2}) \cdot \Delta x_{N-2}\big)  d\pi_{\bar x_{N-2}}(x_{N-1})\\
&\ +\dots + \big( \chi_{2}(x_{1})  - \chi_{2}(x_{2})  + h_{1}(\bar x_{1}) \cdot \Delta x_{1}\big)  d\pi_{\bar x_{1}}(x_{2}) d\pi^1(x_1) \\
& =  \int \sum_{t=1}^{N-1} \big{(}  \chi_{t+1}(x_t)  - \chi_{t+1}(x_{t+1})  + h_{t}(\bar x_t) \cdot \Delta x_t  \big{)} \, d\pi.
\end{align*}
Finally, in \eqref{a8} we can combine the integrals which then yields the claim \eqref{claim1}, hence the theorem.
\end{proof}

\end{document}